\documentclass[lettersize,journal]{IEEEtran}

\ifCLASSINFOpdf
   \usepackage[pdftex]{graphicx}
\else
   \usepackage[dvips]{graphicx}
\fi
\usepackage{url}
\usepackage[cmex10]{amsmath}
\usepackage{amsmath,amssymb,mathrsfs,amsfonts,dsfont}
\usepackage{mathtools}
\usepackage{empheq}
\usepackage{amsthm}
\usepackage[english]{babel}
\newtheorem{theorem}{Theorem}

\newtheorem{remark}{Remark}
\newtheorem{definition}
{Definition}
\newtheorem{assumption}{Assumption}
\usepackage{xcolor}
\usepackage{graphicx}
\usepackage{epsfig}
\usepackage{array}
\usepackage{makecell} 
\usepackage{tabularx}
\allowdisplaybreaks
\begin{document}

\title{An Integer Clustering Approach for Modeling Large-Scale EV Fleets with Guaranteed Performance}

\author{Sijia Geng,\quad \, 
Thomas Lee, \quad
Dharik Mallapragada,\quad 
Audun Botterud\quad
\thanks{Sijia Geng is with the Department of Electrical and Computer Engineering,
Johns Hopkins University, Baltimore, MD 21218, USA.}
\thanks{Thomas Lee is with the MIT Institute for Data, Systems, and Society,
Massachusetts Institute of Technology,
Cambridge, MA 02139, USA.}
\thanks{Dharik Mallapragada is with the MIT Energy Initiative, Massachusetts Institute of Technology, Cambridge, MA 02139, USA.}
\thanks{Audun Botterud is with the MIT Laboratory for Information and Decision Systems, Massachusetts Institute of Technology, Cambridge, MA\,02139, USA.}
\thanks{ 
Corresponding Author: Sijia Geng (email: sgeng@jhu.edu).}
}




\maketitle

\begin{abstract}
Large-scale integration of electric vehicles (EVs) leads to a tighter integration between transportation and electric energy systems. In this paper, we develop a novel integer-clustering approach to model a large number of EVs that manages vehicle charging and energy at the fleet level yet maintain individual trip dispatch. The model is then used to develop a spatially and temporally-resolved decision-making tool for optimally
planning and/or operating EV fleets and charging infrastructure. 
The tool comprises a two-stage framework where a tractable disaggregation step follows the integer-clustering problem to recover an individually feasible solution. 
Mathematical relationships between the integer clustering, disaggregation, and individual formulations are analyzed. We establish theoretical lower and upper bounds on the true individual formulation which underpins a guaranteed performance of the proposed method. The optimality accuracy and computational efficiency of the integer-clustering formulation are also numerically validated on a real-world case study of Boston's public transit network under extensive test instances. Substantial speedups with minimal loss in solution quality are demonstrated.
\end{abstract}


\section{Introduction}
   \label{sec:intro}
The increasing adoption of electric vehicles (EVs) will lead to a tighter integration between transportation and electric energy systems due to the large amount of flexible loads introduced by EVs. EV aggregators and owners of commercial EV fleets can control the charging and operation of the vehicles to respond to time-of-use electricity price signals, achieve peak shaving, or provide ancillary services to the grid \cite{ferc2222}.
While there are extensive studies on the grid impacts of added electrical loads from light-duty EVs in residential settings \cite{BURNHAM2017237,mai2018electrification}, the grid impacts of charging infrastructure for heavy-duty EV fleets are less studied. In this paper, we focus on modeling a large number of commercial EV fleets, such as public buses, school buses, and delivery freight, and develop spatially and temporally-resolved decision-making tools for deploying and operating EVs and the charging infrastructure.

Previous research has looked into the operational problem of managing EV fleets. Reference \cite{al2023multi} proposed a multi-battery flexibility model that uses a few virtual batteries to conservatively estimate the aggregate flexibility set of a large number of EVs. The virtual batteries can be identified through a clustering approach to reflect the various geometric shapes of the individual EV flexibility sets. The formulation of this work is not applicable to planning settings though because the total number and composition of the individual EVs are given parameters. 
Reference \cite{botkin2021distributed} studied the problem of coordinating EV charging 
by distributed control. 
A model for aggregated fleet state-of-charge dynamics was presented. However, the model considered the trip assignments as input parameters instead of decision variables, and the vehicle composition was homogeneous. 
On the other hand, the planning problem for EVs and charging infrastructure has received relatively less attention. 
Reference \cite{pham2022techno} quantified the value of distributed energy resources for meeting the energy demand of heavy-duty EVs, though, taking the electricity load profiles from EV charging as fixed input parameters. 
Reference \cite{li2019mixed} studied an EV fleet scheduling problem incorporating the fleet sizing decision but also uses a heuristic assumption of a fixed charging rule and discretized energy levels. 

To address the gaps in existing literature, in this paper, we develop a novel approach, referred to as the integer-clustering formulation, to efficiently model a large number of EV fleets and their charging and dispatching. 
The formulation clusters the vehicles and chargers by their operating characteristics into a few types, and manages the aggregated energy for each vehicle type and the aggregated charging power by the vehicle-charger type. The model is then used for developing an optimization framework.
The proposed framework optimizes EV charging strategies and dispatching which makes it practically useful for both planning design and as an operational decision-making tool.
The integer-clustering formulation is able to effectively reduce the computational complexity due to the fleet-level management for EV charging. Note that the model is still able to handle vehicle dispatch on the individual vehicle level. Most importantly, from the analytical perspective, we provide theoretical guarantees for the performance of the proposed method by showing a lower and upper bound to the true individual formulation. 
The performance of the proposed framework is demonstrated using a real-world case study based on Boston's public transit bus network. Real geospatial timetable data set for bus schedules and actual cost parameters are used in the case study to provide insightful guidelines for future low-emission electrified transportation systems. 

\section{Integer-Clustering Fleet Formulation}\label{sec:models_integer} 
In this section, we present the novel integer-clustering formulation that models fleet-level EV charging to reduce the computational requirement compared to vehicle-level management while still maintaining individual vehicle dispatch. Similarly, the chargers are managed at the group level, only differentiated by charger types.

The trip schedule demand is assumed to be deterministic. This assumption is viable because we focus on commercial EV fleets in this work. The term ``trip block'', or simply, ``block'', refers to a collection of interconnecting trips which in total begins and ends at the charging station (i.e. depot), and are fulfilled by a single vehicle. Note that all EV chargings take place at the depot.
We assume that there are in total $K$ trip blocks in the schedule table, $I$ vehicle types, $J$ charger types, and $S$ representative days. Within each representative day $s$ there are $T_{\text{d}}$ time intervals with indices $\{\underline{\tau}_s,...,\overline{\tau}_s\}$, and across the concatenated representative days there are in total $T$ time intervals, i.e., $T=S T_{\text{d}}$. 

Some important parameters used in the formulation are introduced here.
$D\in \mathbb{R}^K$ gathers the total travel distance of each trip block. 
$A\in\{0,1\}^{K\times T}$ denotes the en route (1) or idle (0) status for each trip block.
$U\in\{0,1\}^{K\times T}$ gathers the time of leaving the charging station for each block, that is, its value equals 1 at the block’s starting time and 0 elsewhere. 
On the other hand, $V\in\{0,1\}^{K\times T}$ gathers the time of returning to the charging station for each block, that is, its value equals 1 at the time immediately \textit{after} the block’s end time, and 0 elsewhere.
$P_{ij}$ denotes the charging power capacity in kW for type $i$ vehicle using type $j$ charger.
$R_{i}$	is the vehicle energy capacity in kWh for type $i$ vehicle. 

\begin{figure}[t]
\begin{center}
\begin{picture}(245.0, 100.0)
\put(  -7,  0){\epsfig{file=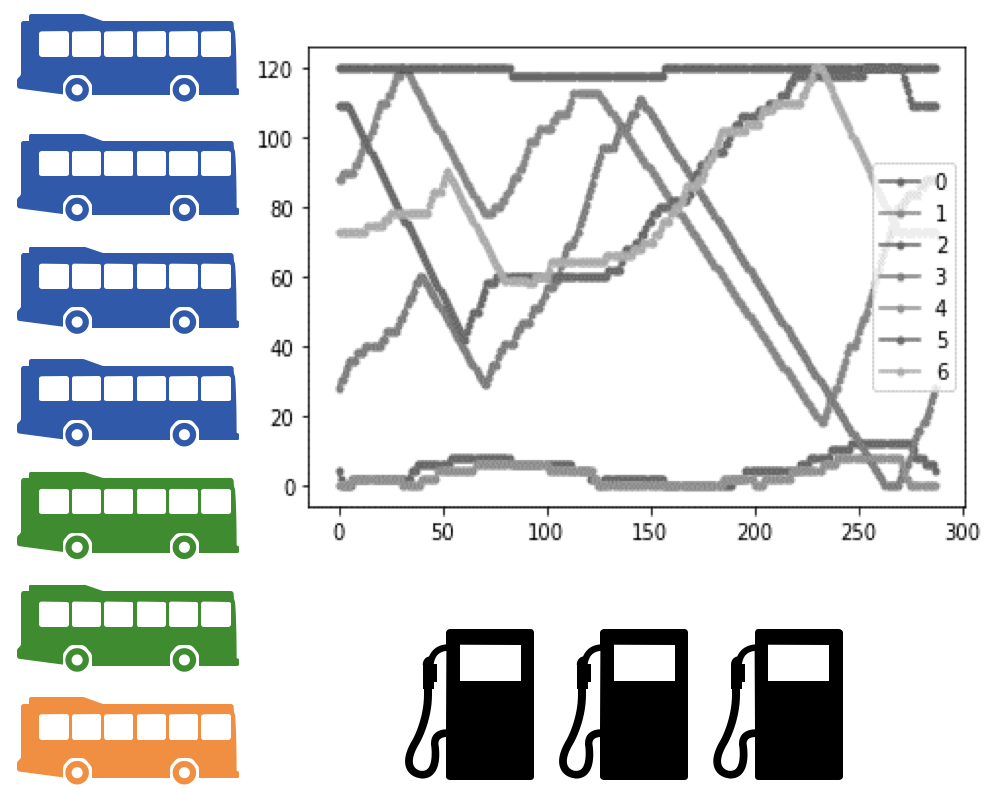,height=.190\textwidth}}  
\put(  124,  0){\epsfig{file=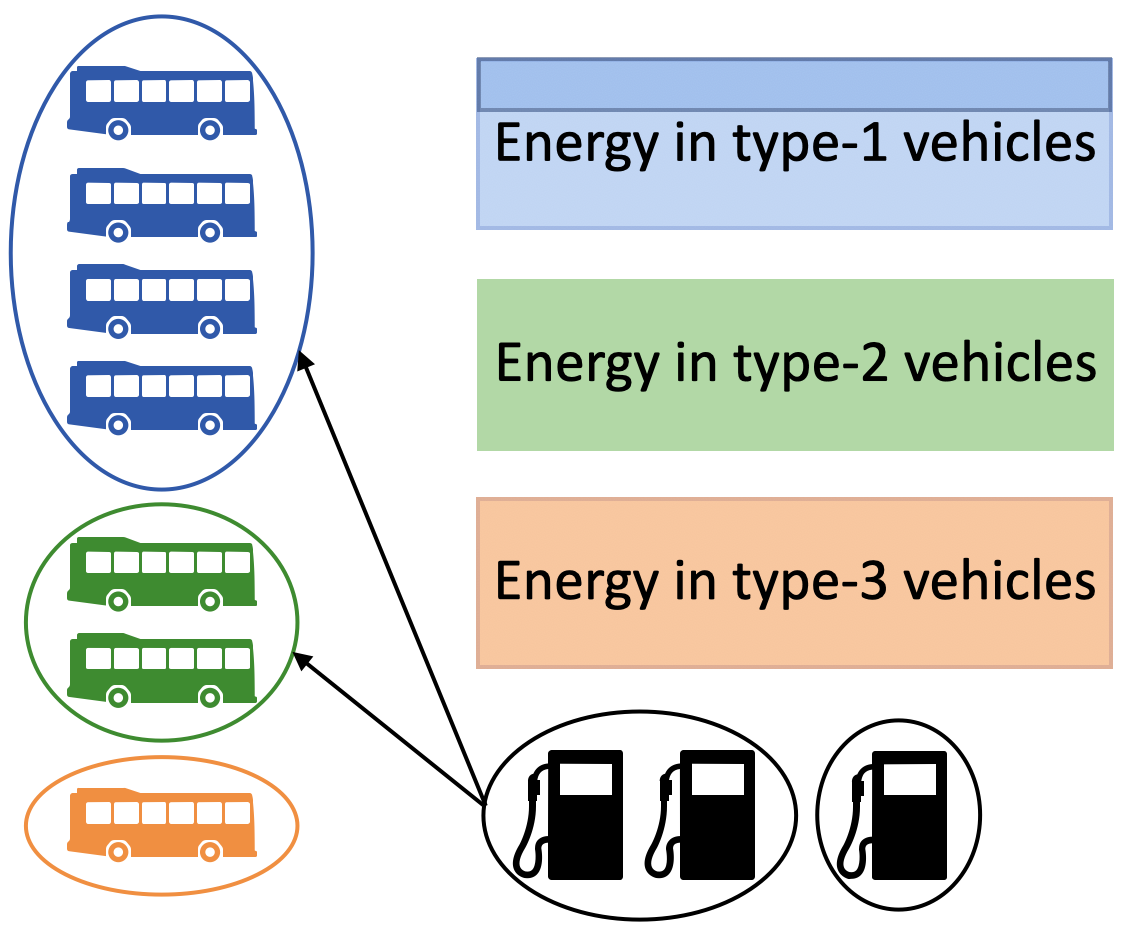,height=.196\textwidth}}  
\put( 105, 5){\small{(a)}}
\put( 235, 5){\small{(b)}}
\end{picture}
\end{center}
      \caption{\small{(a) Conventional formulation: Track charging and state-of-energy for individual vehicles and chargers. (b) Integer-clustering formulation: Group vehicles and chargers by type, and manage charging and SOE at the group level.}}
      \label{fig:integer_clustering}
       \vspace{-0.05cm}
\end{figure}

The fleet-level vehicle dynamics are modeled as follows, 
\begin{align}
&\sum_{i} b_{ki}=1, \ \ b_{ki}\ \text{binary}, \ \forall k, \label{eq_block}\\
  &n_{i}(t) = N_{i}^\text{v}-\sum_{k}(A_{k}(t)b_{ki}),\ \forall i,t,\label{eq_depot}\\
  &\sum_{j} m_{ij}(t)  \leq n_{i}(t),\ \forall i,t, \label{eq_vehicle_depot}\\
   &\sum_{i} m_{ij}(t)  \leq N_{j}^\text{c},\ \forall j,t,\label{eq_vehicle_charging}\\
&n_{i}(t) \geq 0,\ n_{i}(t)\ \text{integer}, \ \ \forall \ i,t,\label{eq_nonneg_n}\\
&m_{ij}(t) \geq 0,\ \forall i,j,t. \ \label{eq_nonneg_m} \\
& N^{\text{v}}_i \geq 0, \ N^{\text{v}}_i\ \text{integer} ,\ \forall i, \ \label{eq_number_Nv}\\
 & N^{\text{c}}_j \geq 0, \ N^{\text{c}}_j\ \text{integer} ,\ \forall j.\label{eq_number_Nc}
\end{align}

Here, $b_{ki}$ is a binary variable that denotes whether block $k$ is assigned to vehicle type $i$, and \eqref{eq_block} ensures that each block is assigned to exactly one vehicle type. \eqref{eq_depot} enforces that $n_{i}$, the number of type $i$ vehicles at the depot, equals $N_{i}^\text{v}$, the total number of type $i$ vehicles, minus those that are en route. The variable $m_{ij}$ denotes the number of type $i$ vehicles that are being charged by type $j$ chargers. 
\eqref{eq_vehicle_depot} requires the number of type $i$ vehicles that are actively charging at the depot to be upper bounded by the number of type $i$ vehicles at the depot.  \eqref{eq_vehicle_charging} requires that the total number of vehicles that are charging using type $j$ chargers is upper bounded by $N_{j}^\text{c}$, the number of type $j$ chargers. 
Non-negativity constraints \eqref{eq_nonneg_n}-\eqref{eq_number_Nc} are enforced. Note that $m_{ij}$ is modeled as a continuous variable to allow for the flexibility of switching the usage of chargers among vehicle types within a time interval.

The following constraints describe fleet-level energy management and individual vehicle dispatch. The variable $soe_{i}$, referred to as the state-of-energy (SOE), denotes the combined energy stored in all type $i$ vehicles that are parked at the depot; $soe_i(t)$ refers to the state at the beginning of interval $t$. $p^{\text{v}}_{i}$ is the combined charging power for type $i$ vehicles at the depot, $d_{ki}$ denotes the energy stored in type $i$ vehicles that is prepared for block $k$, and $\eta^{\text{v}}_i$ is the parameter for driving efficiency in kWh/km for vehicle type $i$. 
\begin{align}
    &soe_{i}(t+1)  = soe_{i}(t)  + p^{\text{v}}_{i}(t)\Delta T +  \nonumber\\
    & \hspace{0.1cm}\sum\nolimits_{k} \big[-U_{k}(t) d_{ki} + V_{k}(t) \big(d_{ki} - D_{k} b_{ki} \eta^{\text{v}}_i \big)\big],\quad \nonumber\\
    & \hspace{0.8cm}\forall \ i,s,t\in [\underline{\tau}_s, \overline{\tau}_s - 1], \label{eq_vehicle_soe}\\
    &0 \leq p^{\text{v}}_{i}(t)  \leq \sum_{j} ( P_{ij} m_{ij}(t)),\ \forall i,t,\label{eq_charge_power}\\
     &soe_{i}(\underline{\tau}_s)  = soe_{i}(\overline{\tau}_s)  + p^{\text{v}}_{i}(\overline{\tau}_s)\Delta T +  \nonumber\\
    & \hspace{0.1cm}\sum\nolimits_{k} \big[-U_{k}(\overline{\tau}_s) d_{ki} + V_{k}(\overline{\tau}_s) \big(d_{ki} - D_{k} b_{ki} \eta^{\text{v}}_i\big)\big]
     ,\ \forall i,s, \label{eq_vehicle_soe_ic}\\
    &0  \leq soe_{i}(t)  \leq R_{i} n_{i}(t),\ \forall i,t. \label{eq_vehicle_soe_limit}
\end{align}
Equation~\eqref{eq_vehicle_soe} states that the variation in the aggregated SOE is driven by the charging power $p^{\text{v}}_{i}$, the decrease caused by a vehicle leaving the depot, and the increase due to a vehicle returning to the depot carrying unused surplus energy. \eqref{eq_charge_power} ensures that $p^{\text{v}}_{i}$ is less or equal to the charging rate times the number of vehicles under each charger type. For the sake of continuous operation, \eqref{eq_vehicle_soe_ic} requires that the combined energy at the end of the representative day returns to its initial value. \eqref{eq_vehicle_soe_limit} is a physical limit that makes sure the aggregated SOE is upper bounded by the vehicle energy capacity times the number of type $i$ vehicles at the depot. 

We are yet to account for the variable $d_{ki}$ in \eqref{eq_vehicle_soe}. Two variants of the formulation are provided. The first formulation allows the vehicle to carry surplus energy, i.e. more than what is needed by the impending trip block,
\begin{align}
    &D_{k} b_{ki} \eta^{\text{v}}_i  \leq d_{ki} \leq R_{i} b_{ki},\ \forall i,k,  \quad \text{(less restrictive).}\label{eq_lower}
\end{align}
In contrast, the second formulation requires the vehicle being dispatched to cover a trip block to carry exactly the same amount of energy needed for the trip, 
\begin{align}
     &D_{k} b_{ki} \eta^{\text{v}}_i  =d_{ki} \leq R_{i} b_{ki},\ \forall i,k,  \quad \text{(more restrictive).}\label{eq_upper}
\end{align}
Note that if $b_{ki} = 0$, that is, type $i$ vehicle is not assigned to cover block $k$, then $d_{ki}$ is forced to be zero.
\begin{remark}
    It is clear that \eqref{eq_upper} is more restrictive than in reality where surplus energy is allowed, while \eqref{eq_lower} is less restrictive than in reality because in this formulation the surplus energy can be redistributed among the vehicle fleet upon a vehicle returning to the depot.
\end{remark}

\section{Individual Vehicles Formulation}\label{sec:models_individual}
Given the decision (of either a planning or operation problem) using the integer-clustering formulation, we are interested in determining: 1) if there exists an exact disaggregation for individual vehicles that, when aggregated, match the charging and dispatching operation determined by the integer-clustering formulation; 2) if so, what is the exact disaggregation and how does it compare to re-optimizing operation (while keeping the investment plan) during disaggregation; and 3) if not, to what extent the integer-clustering formulation matches that of the aggregated charging and dispatching of the individual formulation. 
To facilitate the discussion, we make the following definition.
\begin{definition}
     Consider any pair of trip blocks, indexed by $k_1$ and $k_2$. They are said to be compatible with each other if either: block $k_1$’s ending time is less than or equal to $k_2$’s starting time, or block $k_2$’s ending time is less than or equal to $k_1$’s starting time. \label{definition_compatible} 
\end{definition}

In addition to the variables and parameters introduced earlier, the following sets of variables are defined for type $i$ vehicles. $y^{i}_{v}$ indicates whether the vehicle indexed by $(i,v)$, that is, vehicle $v$ in type $i$, is purchased. $bb^{i}_{kv}$ denotes whether block $k$ is assigned to vehicle $v$ (in type $i$ vehicles), $pp^{\text{v},i}_{vj}(t)\in \mathbb{R}_{+}$ denotes the charging power of vehicle $v$ using a type $j$ charger at time $t$, $dd^{i}_{kv}$ denotes the energy carried in vehicle $v$ for block $k$, and $soev^{i}_{v}(t)\in \mathbb{R}_{+}$ denotes the stored energy of vehicle $v$ at time $t$ while it is at the depot. Besides, define $x^i_{vj}(t)$ as continuous variables indicating the time fraction during interval $t$ when vehicle $v$ is charging using type $j$ charger. 

\begin{remark}
The modeling choice of continuous $x^i_{vj}$ (and continuous $m_{ij}$ in Section~\ref{sec:models_integer}) implicitly allows the same charger to be used across different vehicles and vehicle types within the same time period, thereby mitigating the impact of time resolution choice on model fidelity.
\end{remark}

If a vehicle indexed by $(i,v)$ is assigned for any blocks, then its purchase status must be on and contribute to the fleet size, 
\begin{align} 
    bb^i_{kv} &\leq y^i_{v} \leq 1,\ \ bb^i_{kv},y^i_{v}\ \text{binary}, \ \forall i,k,v, \label{eq_new_indi}\\
    \sum_v y^i_{v} &\leq N^{\text{v}}_i  , \ \forall i,v. \label{eq_new2_indi} 
\end{align}

One vehicle cannot cover two blocks at any time $t$, that is, blocks that are assigned to the same vehicle should be compatible. Recall that matrix $A$ denotes the en route or idle status for each trip block. We have, 
\begin{align}
    &\sum_{k}(A_{k}(t)bb^{i}_{kv})\leq 1,\ \forall i,v,t. \label{eq_compatible_disag}
\end{align}

Vehicles cannot charge en route,
\begin{align}
    \sum_j x^i_{vj}(t) \leq 1 - \sum_k A_k(t) bb^i_{kv}, \  \forall i,v,t. \label{eq_enroute_disag}
\end{align}

During any time period, the total amount of ``charging fraction'' utilized by vehicles must be upper bounded by the number of chargers of type $j$,
\begin{align}
\sum_{i,v} x_{vj}^i(t) \leq N^{\text{c}}_j , \ \forall j. \label{eq_x_p4}
\end{align}

The dynamic of the state-of-energy (while at the depot) for individual vehicles is given by,
\begin{align}
    soev^{i}_{v}(t+1)  &= soev^{i}_{v}(t)  + \sum_{j} pp^{\text{v},i}_{vj}(t) \Delta T +  \nonumber\\
    & \hspace{-0.3cm}  \sum_k \big[-U_{k}(t)
 dd^{i}_{kv} + V_{k}(t) \big( dd^{i}_{kv} - D_{k} bb^{i}_{kv}\eta^{\text{v}}_i\big) \big], \nonumber\\
 & \hspace{1cm}\forall \ i,v,s,t\in [\underline{\tau}_s, \overline{\tau}_s-1].\label{eq_vehicle_soe_disag}
\end{align}

The charging rate and stored energy for individual vehicles should be bounded by their capacity. And the boundary condition is enforced to ensure continuity of operation,
\begin{align}
    &0\leq pp^{\text{v},i}_{vj}(t) \leq P_{ij}x^i_{vj}(t),\ \forall i,j,v,t, \label{eq_charge_power_disag}\\
    &0\leq soev^{i}_{v}(t) \leq R_{i}, \forall i,v,t, \label{eq_soev_disag}\\
    &soev^{i}_{v}(\underline{\tau}_s) = soev^{i}_{v}(\overline{\tau}_s)  + \sum\nolimits_{j} pp^{\text{v},i}_{vj}(\overline{\tau}_s) \Delta T +  \nonumber\\
    & \hspace{0.3cm} \sum\nolimits_{k} \big[-U_{k}(\overline{\tau}_s)
 dd^{i}_{kv} + V_{k}(\overline{\tau}_s) \big( dd^{i}_{kv} - D_{k} bb^{i}_{kv}\eta^{\text{v}}_i\big) \big].
 \label{eq_soev_boundary_disag}
\end{align}
\begin{remark}
    Recall that in the integer-clustering formulation, similar capacity constraints, \eqref{eq_charge_power}-\eqref{eq_vehicle_soe_limit}, were defined at the fleet level for each vehicle type. Those constraints are less restrictive than \eqref{eq_charge_power_disag}-\eqref{eq_soev_boundary_disag} and may represent fictitious ``super charging'' due to the aggregation effect. Therefore, a feasible solution for the integer-clustering formulation may not necessarily have a feasible disaggregation. \label{remark_super_charging}
\end{remark}

We account for the variable $dd^{i}_{kv}$ in \eqref{eq_vehicle_soe_disag} in a similar manner as in the integer-clustering formulation by giving two variants. The first formulation allows the vehicle to carry surplus energy than needed by the impending trip block,
\begin{align}
    &D_{k} bb^{i}_{kv} \eta^{\text{v}}_i  \leq dd^{i}_{kv} \leq R_{i} bb^{i}_{kv},\ \forall i,k,v , \ \text{(less restrictive),}\label{eq_lower_disag}
\end{align}
while the second formulation requires the vehicle being dispatched to cover a trip block to carry exactly the same amount of energy needed for the trip, 
\begin{align}
     &D_{k} bb^{i}_{kv} \eta^{\text{v}}_i  =dd^{i}_{kv} \leq R_{i} bb^{i}_{kv},\ \forall i,k,v,  \ \text{(more restrictive).}\label{eq_upper_disag}
\end{align}

\section{Fleet-Individuals Operational Matching}\label{sec:models_disaggregate}
The following constraints are geared towards finding the exact disaggregation that matches the operational profiles of the integer-clustering solution.
Note that to find the exact disaggregation (as compared to re-optimizing the vehicle operation), there is no interplay across different vehicle types. As a result, the problem can be perfectly decomposed by considering different vehicle types separately. 
The superscript $i$ is added only to indicate that the problem is for type-$i$ fleet.

Assuming that we already obtained the results of the integer-clustering formulation (for either the planning or operation problems), when we disaggregate, every block that is assigned to vehicle type $i$ needs to be covered by one and only vehicle in type $i$,
\begin{align}
    &\sum_{v} bb^{i}_{kv} = b_{ki}, \ \forall k.  \label{eq_block_disag}
\end{align}

The amount of energy carried for block $k$ aggregates to the value assigned to the type-$i$ fleet,
\begin{align}
   \sum_v dd^{i}_{kv}  = d_{ki}, \ \forall k. \label{eq_block_energy_disag}
\end{align}

The number of type $j$ chargers being used for all type $i$ vehicles at any time $t$ should match the total number found in the integer-clustering problem,
\begin{align}
    \sum_v x^i_{vj}(t) = m_{ij}(t), \  \forall j,t.\label{eq_vehicle_charging_disag}
\end{align}  
Note that there is no need to index the chargers; we simply need to utilize them based on their availability, using the first one that's accessible. 

The sum of the charging power for all the type $i$ vehicles should match the total charging power for this vehicle type found in the integer-clustering problem,
\begin{align}
& \sum_{v,j} pp^{\text{v},i}_{vj}(t) = p^{\text{v}}_{i}(t),\ \forall t.\label{eq_charge_power_pp_disag}
\end{align}
Similarly, the sum of the stored energy should match the total aggregated energy for this vehicle type,
\begin{align}
& \sum_{v} soev^{i}_{v}(t) =  soe_{i}(t), \ \forall t.\label{eq_vehicle_soe_limit_disag}
\end{align}

\section{Energy System Model}
The interaction between the charging infrastructure at the depot and the electric distribution grid is described in a simplified manner, as outlined below.  
The electric power balance equation is given by, 
\begin{align}
   &\quad {p}_\text{g}(t) = \sum_{i} p^{\text{v}}_i(t)\ \  (\text{or,} \ {p}_\text{g}(t) = \sum_{v} pp^{\text{v},i}_{vj}(t)), \  \forall t, \label{eq_power_balance_elec}
\end{align}
where ${p}_\text{g}$, the electric power supplied by the grid, equals the total demand from EV charging.  

A typical electricity rate structure for commercial customers consists of energy charge (\$/kWh) and demand charge (\$/kW). The demand charge is based on the peak demand, that is, the highest hourly electricity usage for all time intervals during each billing period. 
We assume that for the selected representative days, there are in total $L$ groups of season and weekday/weekend types with distinct demand charge rates, and $p^\text{pk}_l$ denotes the peak power consumption in the $l$-th group. We have,
\begin{align}
    &p^\text{pk}_l  \geq p_{\text{g}}(t)
    ,\ \forall l, \ \forall t\in \text{group}~l.
\end{align}

Furthermore, the following engineering constraint is enforced, 
\begin{align}
    &0 \leq {{p}}_\text{g}(t) \leq \overline{p}_{\text{g}}, \label{eq_cap_g_e} 
\end{align}
where $\overline{p}_{\text{g}}$ is a parameter that governs the maximum power that can be supplied from the energy distribution system. Alternatively, $\overline{p}_{\text{g}}$ could be treated as a decision variable to model the potential upgrade of the distribution system infrastructure. We do not consider the possibility of feeding energy back to the grid for now, leading to $p_{\text{g}}(t)\geq 0$ in \eqref{eq_cap_g_e}.

\section{Optimization Formulation}\label{sec:opt}
Two types of optimization problems are considered. The planning problem aims to find the least-cost investment portfolio for EV fleets and energy infrastructures that meet electrified transportation demand. The operation problem determines the most cost-effective EV charging and dispatching schedule given the existing fleet and energy infrastructure.

The objective function $J_\text{obj}$ is defined as the summation of investment costs and operational costs,
\begin{align}
    J_\text{obj} & =\sum_i (N^{\text{v}}_i c^{\text{v}}_i)
    +\sum_{j} (N^{\text{c}}_j c^{\text{c}}_j) 
+ \sum_l ( p^{\text{pk}}_{l} c^{\text{pk}}_{l})   
\nonumber&&\\
& \hspace{-0.1cm} + \sum_s \sum_{t\in T_s} (p_{\text{g}}(t) \Delta T c_\text{g}(t)  S_s) + 
 \sum_i  \sum_k D_{k} b_{ki} c^{\text{m}}_i ,\label{eq_objective_plan}
\end{align}
where $c^{\text{v}}_i$, $c^{\text{c}}_j$, $c^{\text{pk}}_{l}$, $c_\text{g}(t)$, and $c^{\text{m}}_i$ are the cost parameters for EVs, chargers, demand charge, energy price, and vehicle maintenance, respectively.
The first two terms in $J_\text{obj}$ are the capital costs for EVs and chargers, which will be a constant in operational problems. $p^{\text{pk}}_{l} c^{\text{pk}}_{l}$ denotes the demand charge and $p_{\text{g}}(t) \Delta T c_\text{g}(t)  S_s$ denotes the energy cost for the electricity supplied by the grid. Note that with a slight abuse of notation, we use $S_s$ to denote how many type $s$ representative days are there, for example, in the time horizon of one year, and $T_s$ to denote the hours in those days. $D_{k} b_{ki} c^{\text{m}}_i$ is the fleet maintenance cost. 

\subsection{Planning Problem}
When determining the investment plan for EV fleets and charging infrastructure, we make the following assumption.
\begin{assumption}
There exists a type of vehicle available for purchase with a range longer than the longest trip block.
Furthermore, there is a type of charger available for purchase that has sufficiently high charging power capacity when paired with the purchased vehicle to charge it to high enough energy levels to satisfy the block schedules.\label{assump_range}
\end{assumption}
These assumptions are reasonable when converting a practical commercial fleet to EVs. Several EV products are available on the market that have a long enough range for public transit or freight delivery needs in an urban setting. Besides, the trip blocks are normally scheduled considering shifts for drivers and allow ample time for charging in between trips. In practice, a significant portion of the charging is performed during the previous night to prepare EVs for the next day's trips. 
Under Assumption~\ref{assump_range}, $\mathcal{P}_1$ and $\mathcal{P}_2$ that will be defined in the following are always feasible due to the trivial solution of procuring a vehicle and charger for every trip block.

\subsubsection{Integer-clustering formulation} 
The planning problem with the integer-clustering formulation is given as follows,
\begin{align}
&\hspace{-1.6cm}  (\mathcal{P}_1) \ \ \min\limits_{\boldsymbol{X}_
\text{agg}}\ J_\text{obj} \nonumber \\
\mathrm{subject\ to}&\quad \text{Fleet operation}\ \  \eqref{eq_block}-\eqref{eq_vehicle_soe_limit},\eqref{eq_lower} \ (\text{or} \ \eqref{eq_upper}),\nonumber\\
&\quad \text{Energy system}\ \ \eqref{eq_power_balance_elec}-\eqref{eq_cap_g_e}.\nonumber
\end{align}
The vector $\boldsymbol{X}_\text{agg}$ contains all decision variables. Investment-related decision variables include $N_i^{\text{v}}$ and $N_j^{\text{c}}$. Operation-related decision variables include $b_{ki}$, $n_{i}(t)$, $m_{ij}(t)$, $p^{\text{v}}_{i}(t)$, $soe_{i}(t)$, $d_{ki}$, ${p}_\text{g}(t)$, $p^\text{pk}_{l}$.

\subsubsection{Individual formulation} 
We consider the planning problem with individual formulation to serve as a benchmark for demonstrating the performance of the integer-clustering formulation. The idea is to incorporate constraints in Sections~\ref{sec:models_integer},~\ref{sec:models_individual}, and \ref{sec:models_disaggregate} in a single optimization problem,
\begin{align}
  (\mathcal{P}_2) \ \ \min\limits_{\boldsymbol{X}_\text{agg},y_v^i, \boldsymbol{X}_\text{indiv}}  & \ J_\text{obj} \nonumber \\
\mathrm{subject\ to} & \nonumber\\
&\hspace{-1cm} \text{Fleet operation}\ \  \eqref{eq_block}-\eqref{eq_vehicle_soe_limit}, \ \eqref{eq_lower}\ (\text{or}\ \eqref{eq_upper}),\nonumber \\
&\hspace{-1cm} \text{Individual vehicles}\ \  \eqref{eq_new_indi}-\eqref{eq_soev_boundary_disag},\ 
\eqref{eq_lower_disag}\ (\text{or}\ \eqref{eq_upper_disag}),\nonumber\\
&\hspace{-1cm} \text{Fleet operations matching}\ \ \eqref{eq_block_disag}-\eqref{eq_vehicle_soe_limit_disag}.\nonumber\\
&\hspace{-1cm} \text{Energy system}\ \ \eqref{eq_power_balance_elec}-\eqref{eq_cap_g_e}.\nonumber
\end{align}
The decision variables are $\boldsymbol{X}_\text{agg}$, $y_v^i$ and $\boldsymbol{X}_\text{indiv}$, where $\boldsymbol{X}_\text{agg}$ is as defined in $\mathcal{P}_1$ and $\boldsymbol{X}_\text{indiv}$ contains $bb_{kv}$, $x^i_{vj}(t)$, $pp^{\text{v}}_{vj}(t)$, $dd_{kv}$, and $soev_{v}(t)$. Note for $\mathcal{P}_2$, for each type $i$ the vehicle index $v$ ranges from 1 to $K$, since $K$ is the ex-ante upper bound on how many vehicles may need to be purchased.

\subsection{Operation Problem}
Given an optimal solution $\boldsymbol{X}^{*}_\text{agg,1}$ from the integer-clustering planning problem $\mathcal{P}_1$, for the sake of finding an exact individually-feasible disaggregation, or showing the lack thereof, we consider the feasibility problem $\mathcal{P}_3$,
\begin{align}
&\hspace{-1.6cm}  (\mathcal{P}_3) \ \ \text{find} \ \ \ {\boldsymbol{X}_\text{agg}},{\boldsymbol{X}_\text{indiv}}\ \nonumber \\
\mathrm{subject\ to} \nonumber\\
&\hspace{-1cm} \text{Planning coupling}\ \ \boldsymbol{X}_\text{agg} = \boldsymbol{X}^{*}_\text{agg,1},\nonumber\\
&\hspace{-1cm}
\text{Individual vehicles}\ \  \eqref{eq_compatible_disag}-\eqref{eq_soev_boundary_disag},\ 
\eqref{eq_lower_disag}\ (\text{or}\ \eqref{eq_upper_disag}),\nonumber\\
&\hspace{-1cm} \text{Fleet operations matching}\ \ \eqref{eq_block_disag}-\eqref{eq_vehicle_soe_limit_disag}.\nonumber\\
&\hspace{-1cm} \text{Energy system}\ \ \eqref{eq_power_balance_elec}-\eqref{eq_cap_g_e},\nonumber
\end{align}
where the fleet-related decision variables $\boldsymbol{X}_\text{agg}$ (c.f. $\mathcal{P}_1$) are fixed to their respective optimal solution values from $\mathcal{P}_1$. The vector $\boldsymbol{X}_\text{indiv}$ contains individual vehicles' operational decision variables (c.f. $\mathcal{P}_2$). For $\mathcal{P}_3$, the vehicle index set is constructed according to the aggregated sizing solution, namely type $i$'s vehicle index $v$ ranges from 1 to $N^{\text{v}*}_i$. Moreover, by construction, equations \eqref{eq_new_indi} and \eqref{eq_new2_indi} are satisfied since there are exactly $N^\text{v}_i$ variables constructed.

Instead of strictly following the aggregated operational charging profile, we can consider the disaggregation optimization problem, referred to as $\mathcal{P}_4$, to determine the optimal operation for individual vehicles while fixing only the planning investments (for vehicles and chargers) and vehicle-type block assignments according to those decided by $\boldsymbol{X}_{\text{agg},1}^*$,
\begin{align}
&\hspace{-1.6cm}  (\mathcal{P}_4) \ \ \min\limits_{\boldsymbol{X}_\text{agg}, \boldsymbol{X}_\text{indiv}}   \  J_\text{obj} \nonumber \\
\mathrm{subject\ to} & \nonumber\\
&\hspace{-1cm} \text{Planning coupling}\ \ \{N^{\text{v}},N^{\text{c}},b\} = \text{Proj}(\boldsymbol{X}_{\text{agg},1}^*) \nonumber\\
&\hspace{-1cm} \text{Individual vehicles}\ \  \eqref{eq_compatible_disag}-\eqref{eq_soev_boundary_disag},\ 
\eqref{eq_lower_disag}\ (\text{or}\ \eqref{eq_upper_disag}),\nonumber\\
&\hspace{-1cm} \text{Fleet operations matching}\ \ \eqref{eq_block_disag}-\eqref{eq_vehicle_soe_limit_disag}.\nonumber\\
&\hspace{-1cm} \text{Energy system}\ \ \eqref{eq_power_balance_elec}-\eqref{eq_cap_g_e},\nonumber
\end{align}
where ``Proj'' is a function that projects the vector $\boldsymbol{X}_{\text{agg},1}^*$ to its components $N_i^\text{v}$, $N_j^\text{c}$, $b_{ki}$, which are now fixed constants in $\mathcal{P}_4$. Similarly to $\mathcal{P}_3$, the vehicle index is constructed to be from 1 to $N^{\text{v}*}_i$. Note that for $\mathcal{P}_4$, other than the fixed vehicle assignments, the remaining $\boldsymbol{X}_{\text{agg}}$ decisions consist of trivial aggregations $d,m,p,soe$ which are no longer coupled to the $\mathcal{P}_1$ solution. Such a freedom to deviate from the $\mathcal{P}1$ charging schedules implies that $\mathcal{P}4$'s feasible region contains $\mathcal{P}3$'s feasible region.

\section{Theoretical Bounds on Performance}
In this section, we provide theoretical analysis and establish performance guarantees for the proposed method. We give both an upper and lower bound for the true individual formulation using the integer-clustering model.
\begin{theorem} 
\label{Th1} 
Let $\underline{J}_1^*$ and $\overline{J}_1^*$ be the optimal values of problem $\mathcal{P}_1$ with constraint \eqref{eq_lower} and \eqref{eq_upper}, respectively. Let $\underline{J}_2^*$ and $\overline{J}_2^*$ be the optimal values of problem $\mathcal{P}_2$ with constraints \eqref{eq_lower},\eqref{eq_lower_disag} and \eqref{eq_upper},\eqref{eq_upper_disag}, respectively. 
Let $\underline{J}_4^*$ and $\overline{J}_4^*$ be the optimal values of problem $\mathcal{P}_4$ with constraint \eqref{eq_lower_disag} and \eqref{eq_upper_disag}, respectively. 
Then $\underline{J}_1^* \leq \underline{J}_2^* \leq \underline{J}_4^*$ and $\overline{J}_1^* \leq \overline{J}_2^* \leq \overline{J}_4^*$.

Moreover, consider the disaggregation feasibility problem $\mathcal{P}_3$ that corresponds to $\mathcal{P}_1$. If $\mathcal{P}_3$ is feasible, then $\underline{J}_1^*=\underline{J}_2^*=\underline{J}_4^*$ and $\overline{J}_1^* =\overline{J}_2^*=\overline{J}_4^*$.

Finally, $\underline{J}_1^* \leq \overline{J}_1^*$, $\underline{J}_2^* \leq \overline{J}_2^*$, and $\underline{J}_4^* \leq \overline{J}_4^*$.
\end{theorem}

\begin{proof} \label{Pf1} 
Note that $\mathcal{P}_2$ is always feasible under Assumption~\ref{assump_range}.
Let $\hat{\boldsymbol{X}}_2=$($\hat{\boldsymbol{X}}_\text{agg}$, $\hat{y}_v^i$, $\hat{\boldsymbol{X}}_\text{indiv})$ be a feasible solution of $\mathcal{P}_2$ with constraint \eqref{eq_lower},\eqref{eq_lower_disag} (respectively, \eqref{eq_upper},\eqref{eq_upper_disag}). 
Define $\hat{\boldsymbol{X}}_1=\hat{\boldsymbol{X}}_\text{agg}$.
Because the constraints of $\mathcal{P}_2$ include those of $\mathcal{P}_1$, we know that $\hat{\boldsymbol{X}}_1$ is a feasible solution of $\mathcal{P}_1$ with constraint \eqref{eq_lower} (respectively, \eqref{eq_upper}). Therefore, $\underline{J}_1^* \leq \underline{J}_2^*$ and $\overline{J}_1^* \leq \overline{J}_2^*$.

To show that $\underline{J}_2^* \leq \underline{J}_4^*$ and $\overline{J}_2^* \leq \overline{J}_4^*$, we analyze the constraints in $\mathcal{P}_2$ and $\mathcal{P}_4$. First, if $\mathcal{P}_4$ is infeasible, these inequalities automatically hold. When $\mathcal{P}_4$ is feasible, 
we first notice that adding constraints \eqref{eq_block}-\eqref{eq_vehicle_soe_limit},\eqref{eq_lower}(\text{or} \eqref{eq_upper}) 
to $\mathcal{P}_4$ (with design parameters $N^{\text{v}},N^{\text{c}},b$ fixed to the corresponding values in $\hat{\boldsymbol{X}}_1$) does not reduce and therefore does not change the feasible region of $\mathcal{P}_4$.
Besides, adding constraints \eqref{eq_new_indi} and \eqref{eq_new2_indi} to $\mathcal{P}_4$ does not reduce and therefore does not change the feasible region of $\mathcal{P}_4$ either, because $N^{\text{v}}_i$ is a parameter and \eqref{eq_new_indi} and \eqref{eq_new2_indi} are trivially satisfied in this case. 
We denote the revised equivalent problem that contains the additional constraints as $\mathcal{P}_4'$. 
Comparing back to $\mathcal{P}_2$, it is clear that they have the same set of constraints except that $\mathcal{P}_4'$ has an additional constraint that sets $N^{\text{v}},N^{\text{c}},b$ to be constants. As a result, take any feasible solution of $\mathcal{P}_4'$ (thereby of $\mathcal{P}_4$), it is also feasible for $\mathcal{P}_2$, which gives $\underline{J}_2^* \leq \underline{J}_4^*$ and $\overline{J}_2^* \leq \overline{J}_4^*$.

Next, we prove the second statement on the equality result. 
Under Assumption~\ref{assump_range}, problem $\mathcal{P}_1$ is always feasible. Let $\tilde{\boldsymbol{X}}_1$ be a feasible solution of $\mathcal{P}_1$ with constraint \eqref{eq_lower} (respectively, \eqref{eq_upper}). Consider the corresponding disaggregation problem $\mathcal{P}_3$. When $\mathcal{P}_3$ is feasible, denote its feasible solution by $\tilde{\boldsymbol{X}}_3=(\tilde{\boldsymbol{X}}_1,{\tilde{\boldsymbol{X}}_\text{indiv}})$. Further define $y^{i}_{v}=1$ if $1\leq v\leq \tilde{N}^{\text{v}}_i$ and $y^{i}_{v}=0$ for $\tilde{N}^{\text{v}}_i < v \leq K$. 
Then
$(\tilde{\boldsymbol{X}}_1,y^{i}_{v},{\tilde{\boldsymbol{X}}_\text{indiv}})$ would be a feasible solution to $\mathcal{P}_2$. As a result, $\underline{J}_2^* \leq \underline{J}_1^*$ and $\overline{J}_2^* \leq \overline{J}_1^*$. Therefore, $\underline{J}_2^* = \underline{J}_1^*$ and $\overline{J}_2^* = \overline{J}_1^*$.

Define $\tilde{\boldsymbol{X}}_4=\tilde{\boldsymbol{X}}_3$, it is obvious that $\tilde{\boldsymbol{X}}_4$ is feasible for $\mathcal{P}_4$. As a result, $\underline{J}_4^* \leq \underline{J}_1^*$ and $\overline{J}_4^* \leq \overline{J}_1^*$. Therefore, $\underline{J}_4^* = \underline{J}_1^*$ and $\overline{J}_4^* = \overline{J}_1^*$. The result is proven.

Finally, to show that $\underline{J}_1^* \leq \overline{J}_1^*$, $\underline{J}_2^* \leq \overline{J}_2^*$, and $\underline{J}_4^* \leq \overline{J}_4^*$, we simply need to notice that the formulations with constraints \eqref{eq_lower} and \eqref{eq_lower_disag} are less restrictive than their counterparts with constraints \eqref{eq_upper} and \eqref{eq_upper_disag}, respectively. 
\end{proof}

Since constraints in the integer-clustering formulation for the aggregated vehicle state-of-energy and charging power are less restrictive compared to the individual formulation, there might not exist a feasible decomposition for individual vehicles given a result from the integer-clustering formulation. However, we always have both an upper bound (by $\mathcal{P}_4$, upon adding slack variables to investment) and a lower bound (by $\mathcal{P}_1$) for the individual formulation.
If however a feasible decomposition exists, the optimal result from the integer-clustering formulation gives the optimal solution for the individual formulation. We can exploit this property to take advantage of the computational efficiency of the integer-clustering formulation yet obtain the true optimal solution that would have been produced by the individual formulation (which may require unrealistic computation time or computing power to solve).

\section{Case Study and Numerical Results}\label{sec:numerical}
The performance and computational efficiency of the developed tools are demonstrated through a real-world application to the public transit bus system of the city of Boston. 

The optimization models were implemented in JuMP using Julia 1.9.2 and solved using Gurobi 10.0. The numerical experiments were conducted on the MIT SuperCloud using Intel Xeon
Platinum 8260 processors each with 48 physical cores \cite{reuther2018interactive}.

\subsection{Fleet Data: MBTA Transit Schedule}\label{subsec:fleet_data}
Transit agencies publish detailed schedule information using the General Transit Feed Specification (GTFS). We use schedules published by the Massachusetts Bay Transportation Authority (MBTA) \cite{mbta-gtfs}, focusing on the Cabot bus depot which serves major city routes. Figure~\ref{fig:cabot_schedules_a}\,(a) shows the geographical location of the Cabot depot and the routes supported by the depot. Figure~\ref{fig:cabot_schedules_a}\,(b) demonstrates the block schedules for the depot on a representative weekday in the Fall season. 
The blue lines indicate the start and end hours of each trip block. It can be seen that the block schedules are complex and represent operation timescale information that should be included in the investment planning. 

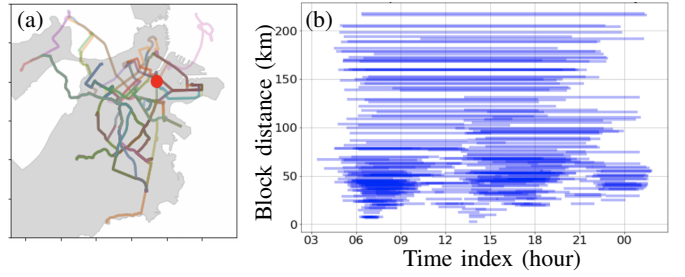
\begin{figure}[t]
\begin{center}
\begin{picture}(245.0, 95.0)
\put(  -7,  0){\epsfig{file=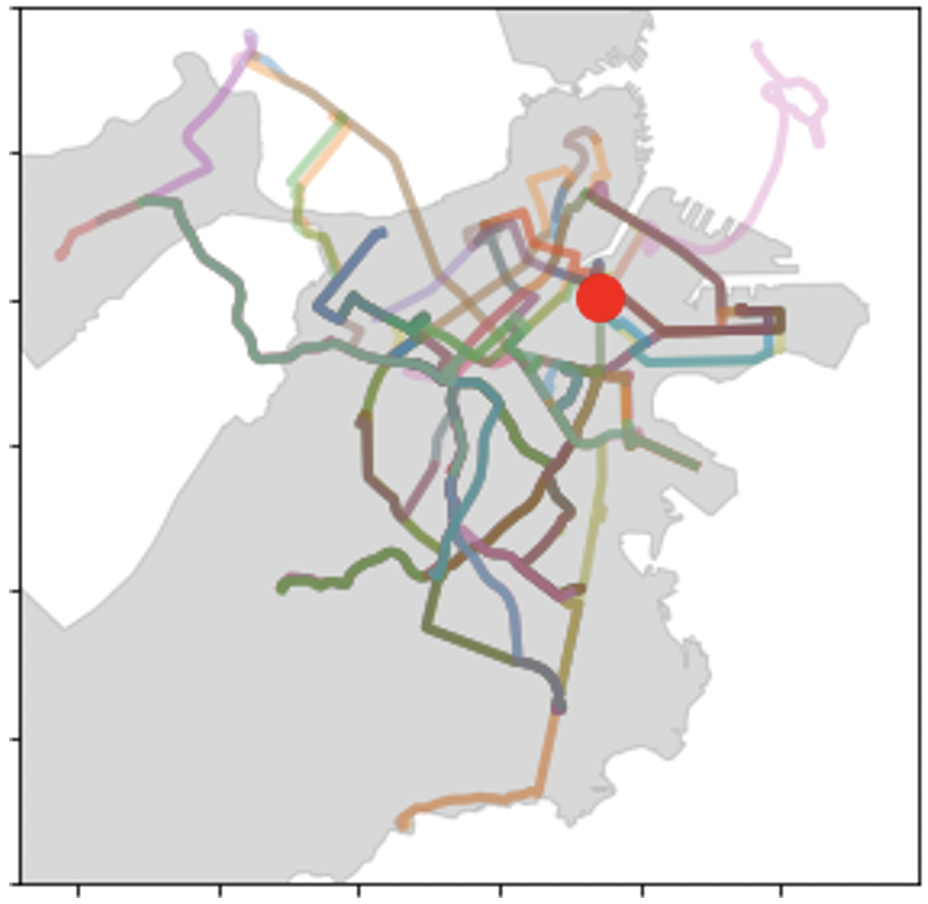,height=.176\textwidth,width=.170\textwidth}}      
\put(  94,  -2){\epsfig{file=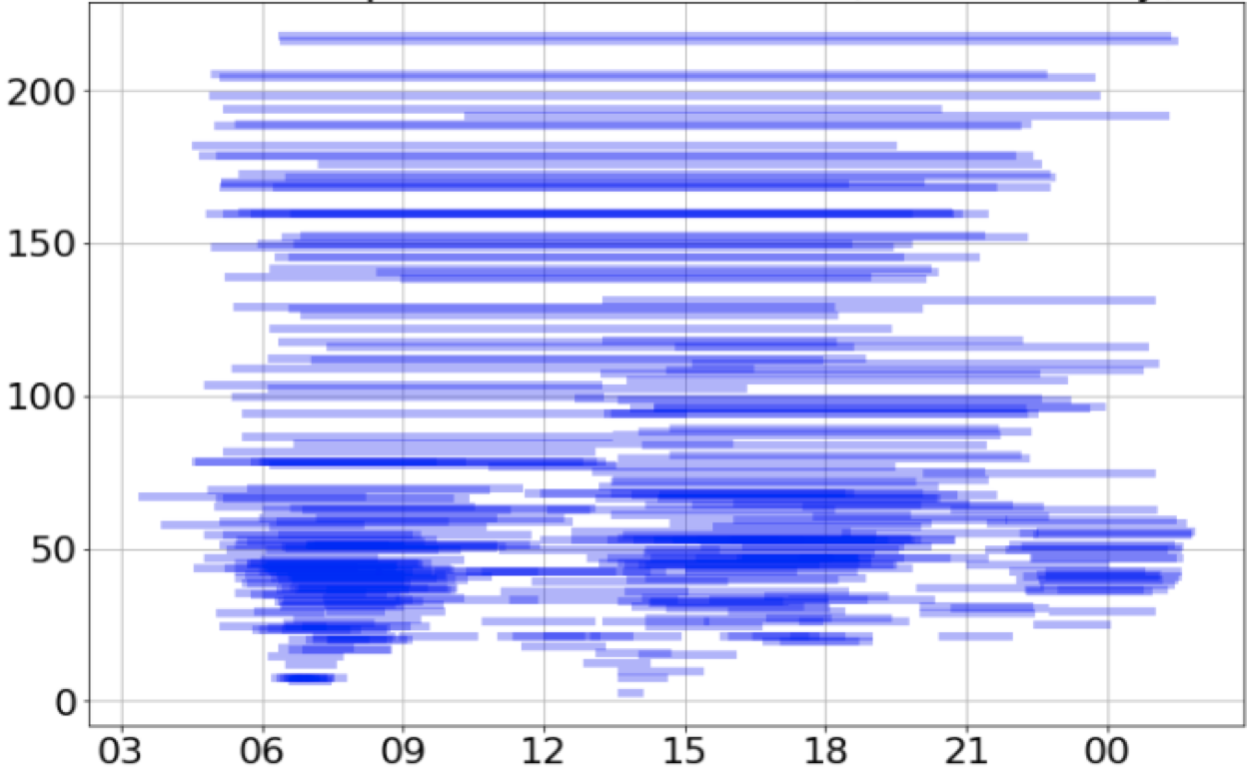,height=.180\textwidth}}  

\put( 87, 8){\small\rotatebox{90}{Block distance (km)}}
\put( 143, -9){\small{Time index (hour)}}

\put( -3, 81.5){\small{(a)}}
\put( 106, 81.5){\small{(b)}}
\end{picture}
\end{center}
      \caption{\small{(a) Routes supported by the Cabot bus depot. (b) Block schedules for the Cabot bus depot on a representative weekday in the Fall season.}}\label{fig:cabot_schedules_a}
       \vspace{-0.05cm}
\end{figure}

\subsection{Parameters for EV, Charger, and Energy System}\label{subsec:vehicle_para} 
In this study, we consider two EVs from the Proterra ZX5 product line, featuring short range and long range, respectively. Their economic and technical parameters are summarized in Table~\ref{table:EV_parameter}.
Three types of DC fast chargers with different charging levels are evaluated too (Table~\ref{table:charger_parameter}).

The retail electricity energy price is assumed to be 13.2~cents/kWh. Besides, a monthly demand charge structure is enforced, with a higher demand charge in the summer months from June to September (24.09~\$/kW) versus the other months (17.92~\$/kW) \cite{openei}.

\setlength\tabcolsep{0pt}
\begin{table*}[t] 
\centering
\caption{Economic and technical parameters of EVs \cite{hgacbuy}.}\label{table:EV_parameter}
\begin{tabular}{c || c c c c c c }
\Xhline{2\arrayrulewidth}\\[-0.9em]
\small{\ \ Vehicle Type}  \ & \small{\ Energy capacity\ \ } & \small{\ Range\ }\  & \small{\ \ Capital cost\ \ } & \small{\ Maintenance cost\ \ } & \small{\ Full charge time\ \ } & \small{\ Life time\ } \\[-0.9em]\\
\Xhline{2\arrayrulewidth}\\[-0.6em]
\small{Proterra ZX5 BEV (short-range)}\  & 225 kWh & 106 miles & \$800,000 & \$0.64/km & 0.45 - 4.5 hours & 12 years  \\[-1em]\\ \hline\\[-0.6em]
\small{Proterra ZX5+ BEV (long-range)} \ & 450 kWh & 197 miles  & \$821,944 & \$0.64/km & 0.45 - 4.5 hours  & 12 years  \\[-1em]\\
\Xhline{2\arrayrulewidth}
\end{tabular}
\end{table*}

\setlength\tabcolsep{0pt}
\begin{table*}[t] 
\centering
\caption{Economic and technical parameters of DC fast chargers\cite{ kushwah2021techno}.}\label{table:charger_parameter}
\begin{tabular}{c || c c c c }
\Xhline{2\arrayrulewidth}\\[-0.9em]
\small{\ \ Charger Type}  \ & \small{\ Power rating\ \ } & \small{\ Capital cost\ }\  & \small{\ \ Installation cost\ \ } & \small{\ Life time\ } \\[-0.9em]\\
\Xhline{2\arrayrulewidth}\\[-0.6em]
\small{\ Level 3 DC fast charger \ } & 50 kW & \$37,000 & \$22,626 & 28 years  \\[-1em]\\ \hline\\[-0.6em]
\small{\ Level 4 DC fast charger\ } & 150 kW & \$45,000 & \$22,626 & 28 years  \\[-1em]\\ \hline\\[-0.6em]
\small{\ Level 5 DC fast charger\ } & 500 kW & \$349,000 & \$250,000 & 28 years  \\[-1em]\\ 
\Xhline{2\arrayrulewidth}
\end{tabular}
\end{table*}

\subsection{Experimental Setup}
 We study a setting of $I=2$ (short- and long-range EV types), $J=3$ (low-, medium-, and high-speed chargers). Data subsamples of varying sizes were generated to demonstrate the proposed method's scalability. For subsample windows $\omega$ ranging from 500 to 3, we select every $\omega$-th block (arranged chronologically), with the largest subsample data containing $K=590$ blocks. Larger subsets were not pursued in this study due to the benchmark $\mathcal{P}_2$ model's intractability.
 
 For each test case, we compare the performance of the individual model (i.e., $\mathcal{P}_2$) versus the proposed integer-clustering plus disaggregation method (i.e., $\mathcal{P}_1, \mathcal{P}_3, \mathcal{P}_4$). All problems are mixed integer linear programming (MILP).
 \\
 \\
\textbf{Individual model} $\mathcal{P}_2$: The individual model is implemented using a vehicle index $v$ (for each vehicle type $i$) ranging from 1 to the aforementioned upper bound $K$.
\\
\textbf{Cluster-disaggregate approach}: 
\begin{itemize}
    \item Solve  $\mathcal{P}_1$, and save optimal solution $\boldsymbol{X}_{\text{agg},1}^*$.
    \item Fix $\boldsymbol{X}_{\text{agg},1}^*$ to solve  $\mathcal{P}_3$
    \item Fix $\boldsymbol{X}_{\text{agg},1}^*$ to solve  $\mathcal{P}_4$.
\end{itemize}

\subsection{{Results: Accuracy of Integer-Clustering Formulation}}

\subsubsection{\underline{Feasible $\mathcal{P}_3$  implies 0 gap}} In Table \ref{table:p3_feas}, a case of $K=24$ resulted in a feasible $\mathcal{P}_3$; simultaneously, the $\mathcal{P}_2$ planning solution and objective exactly matches that of $\mathcal{P}_1$. The objectives of the optimization problems exaclty match: $J_1^* = J_2^* = J_4^*$. Figure \ref{fig:k_vs_gap} extends across all cases, and shows that when the $\mathcal{P}_3$ disaggregation is feasible (indicated by red points on the plot), the resulting optimal value exactly matches $J_1^* = J_2^* = J_4^*$ (the latter equality was also validated for these cases). This numerically validates the equality case of Theorem 1.
\\
\subsubsection{\underline{Optimality gaps are reasonably small}} Regardless of the $P_3$ feasibility problem status, we can attempt to solve the $P_4$ disaggregation optimization problem. Across all cases, the $J_4^*/J_1^*$ gap, which is an upper bound to the $J_2^*/J_1^*$ gap, were empirically found to be within 0.5\% as seen in Figure \ref{fig:k_vs_gap}. Even when the exact $P_3$ is infeasible, such as reported in Table \ref{table:p3_infeas}, $P_4$ can be calculated at relatively low computational burden in order to certify that the obtained individually-feasible solution is within 0.007\% of optimality. In this case, it turns out that $J_1^*$ also had a zero gap to the true optimal value $J_2^*$.  
\\
\subsubsection{\underline{Non-monotonic trend in optimality gaps}} We observe in Figure \ref{fig:k_vs_gap} an overall non-monotonic trend in the accuracy gap as the block number increases: below 8 blocks, the aggregation is always exact; at medium number of blocks there are cases with larger gaps up to 0.4\% (although most cases have low gaps as seen in the dot density); finally as trip blocks increase further (beyond $\sim$100 blocks) there is an overall trend of tightening accuracy. Locally weighted scatterplot smoothing (LOWESS) is used to visualize this trend. At a single vehicle, the aggregated model is trivially equivalent to the individual model. Future work could more closely examine the drivers for increased accuracy of integer clustering at larger block sizes.

\begin{table}[t]
\centering
\caption{$K=24$ ($\omega=76$; allow surplus charging)}
\begin{tabular}{c || c    c c c}
\Xhline{2\arrayrulewidth}\\[-0.9em]
\ \ \ Problem\ \ \ &\ \  \ \ \ \ \ $\mathcal{P}_2$ indiv.\ \ \ \ \ \ \ \ \ & \ \ \ \ $\mathcal{P}_1$ agg.\ \ \ \ \ &\  $\mathcal{P}_3$ feas.\ & \ \ \ $\mathcal{P}_4$ disagg.\ \ \  \\[-0.9em]\\
\Xhline{2\arrayrulewidth}\\[-0.6em]
Feasible & Yes & Yes & Yes & Yes \\[-1em]\\ \hline\\[-0.6em]
$[N^{\text{v}}]$ & $1,1$ & $1,1$ & -- & -- \\[-1em]\\ \hline\\[-0.6em]
$[N^{\text{c}}]$ & $0,1,0$ & $0,1,0$ & -- & -- \\[-1em]\\ \hline\\[-0.6em]
$[p^{\text{pk}}]$ & 71.21, \ 74.48 & 71.21, \ 74.48 & -- & 71.21, \ 74.48 \\[-1em]\\ \hline\\[-0.6em]
Optimal $J^*$ & \$154,526.31 & \$154,526.31 & -- & \$154,526.31 \\[-1em]\\ \hline\\[-0.6em]
Gap vs. $J_2^*$ & -- & 0.000\% & -- & 0.000\% \\[-1em]\\ \hline\\[-0.6em]
Time (sec) & 15.289 & 0.909 & 0.021 & 0.065\\[-1em]\\ \Xhline{2\arrayrulewidth}
\end{tabular} \label{table:p3_feas}
\end{table}

\begin{table}[t]
\centering
\caption{$K=85$ ($\omega=21$; allow surplus charging)}
\begin{tabular}{c || c    c c c}
\Xhline{2\arrayrulewidth}\\[-0.9em]
\ \ \ Problem\ \ \ &\ \  \ \ \ \ \ $\mathcal{P}_2$ indiv.\ \ \ \ \ \ \ \ \ & \ \ \ \ $\mathcal{P}_1$ agg.\ \ \ \ \ &\  $\mathcal{P}_3$ feas.\ & \ \ \ $\mathcal{P}_4$ disagg.\ \ \  \\[-0.9em]\\
\Xhline{2\arrayrulewidth}\\[-0.6em]
Feasible & Yes & Yes & No & Yes \\[-1em]\\ \hline\\[-0.6em]
$[N^{\text{v}}]$  & $7,1$ & $7,1$ & -- & -- \\[-1em]\\ \hline\\[-0.6em]
$[N^{\text{c}}]$ & $0,1,0$ & $0,1,0$ & -- & -- \\[-1em]\\ \hline\\[-0.6em]
$[p^{\text{pk}}]$ &  45.48, \ 65.52 & 45.48,\ 65.52 & -- & 47.28,\ 65.52 \\[-1em]\\ \hline\\[-0.6em]
Optimal $J^*$ &  \$586,713.74 & \$586,713.74 & -- & \$586,757.02 \\[-1em]\\ \hline\\[-0.6em]
Gap vs. $J_2^*$ & -- & 0.000\% & -- & 0.007\%\\[-1em]\\ \hline\\[-0.6em]
Time (sec) & 1,847.439 & 0.714 & 0.446 & 0.945\\[-1em]\\ \Xhline{2\arrayrulewidth}
\end{tabular} \label{table:p3_infeas}
\end{table}

\begin{figure}[t]
\begin{center}
\begin{picture}(300.0, 180.0)
\epsfig{file=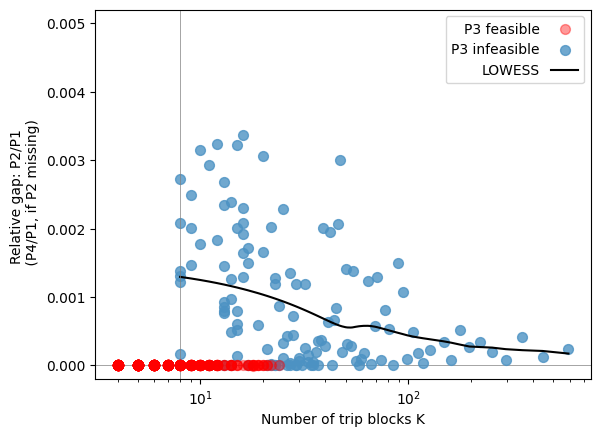,width=.49\textwidth}  
\end{picture}
\end{center}
      \caption{\small{\textbf{Empirical accuracy}: Relative gap of the integer-clustering optimal value $J_1^*$ compared to the true individual optimal value $J_2^*$. If $\mathcal{P}_2$ did not solve within the set 1-hour time limit, the relative gap of $J_4^*$ (which is a valid upper bound) versus $J_1^*$ is reported instead. (Showing cases when surplus energy is allowed.) An outlier of 0.5\% at $K=137$ is omitted. Note that 8 of the $\mathcal{P}_4$ cases were infeasible; but these were solved after allowing an $N^\text{c}$ slack $\leq$ 1.}}
      \label{fig:k_vs_gap}
       \vspace{-0.05cm}
\end{figure}

\subsection{{Results: Computational Efficiency}}
In the Table \ref{table:p3_infeas} case, the individual model took 1,113 times longer than the combined times of $(\mathcal{P}_1,\mathcal{P}_4)$. Figure \ref{fig:k_vs_time} summarizes this speedup across all cases: it shows how the computation time for $\mathcal{P}_2$ scales roughly exponentially with the case block dimension. Meanwhile, the $\mathcal{P}_4$ time (which dominates time needed for the integer-clustering $\mathcal{P}_1$) only increases modestly through the cases, although an uptick in the largest case examined is noticeable. The speedup factor (when using $\mathcal{P}_1,\mathcal{P}_4$) can reach as high as 2000 times, suggesting practical benefits of adopting such an integer-clustering approach. Notably, beyond about $K=100$, $\mathcal{P}_2$ models were not able to solve within 1 hour; in contrast, in these cases, the integer clustering plus disaggregation approach could reliably produce an individually feasible solution with provable optimality guarantees.

\begin{figure}[t]
\begin{center}
\begin{picture}(300.0, 200.0)
\epsfig{file=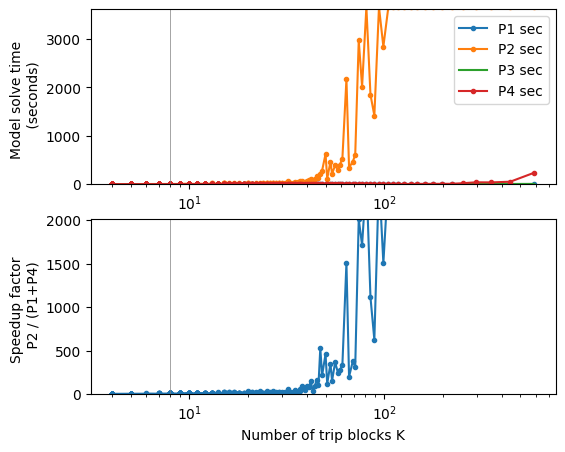,width=.49\textwidth}  
\end{picture}
\end{center}
      \caption{\small{\textbf{Computational efficiency}: Model solution time and speedup factors $T_{\mathcal{P}_2} / (T_{\mathcal{P}_1} + T_{\mathcal{P}_4})$. $\mathcal{P}_2$ models were not reliably solved within the 1-hour time limit at cases larger than $K\sim 100$, and speedup factors are not reported for these cases.}}
      \label{fig:k_vs_time}
       \vspace{-0.05cm}
\end{figure}

\section{Conclusion}
This paper proposed an integer-clustering formulation that can efficiently model a large number of EVs to address the computational challenges in integrated grid-transportation studies. The model was then used to develop decision-making tools for optimally planning and/or operating EV fleets and charging infrastructure. 
We conducted a rigorous analysis to reveal the mathematical relationships between the hierarchical formulations of the integer-clustering, disaggregation, and individualized problems. Most importantly, we proved both lower and upper bounds for the true individual solution based on the proposed integer-clustering formulation. 
These theoretical bounds provide guaranteed performance for the proposed method.

We applied the proposed framework to a real-world case study of Boston's public transit network using real geospatial data for bus schedules and actual cost parameters for EVs and charging infrastructure. The proposed integer-clustering approach demonstrated remarkable accuracy and computational performance. When solving real-life scale problem instances, the individual-vehicle formulation using the full set of potential combinatorial variables is shown to be computationally intractable. In contrast, the integer-clustering plus disaggregation method scales well to larger cases, with documented speedups up to 2000x. This computational efficiency is paired with the quantified optimality gap guarantee, which empirically proves to be within 0.5\% for the instances studied.

Future work includes more deeply examining the mechanisms that cause $\mathcal{P}_3$ disaggregation infeasibilities, and in turn, better capturing these constraints in the integer-clustering formulation.

\bibliographystyle{IEEEtran}
\bibliography{RefTransp}


 




\vfill

\end{document}